\newif\ifFull
\Fullfalse

\ifFull
\documentclass[11pt]{article}
\else
\documentclass[runningheads]{llncs}
\fi

\usepackage{graphicx}
\usepackage{amsmath}
\ifFull
\usepackage{amsthm}
\usepackage{fullpage}
\usepackage{cramped}
\fi
\usepackage{amssymb}
\usepackage{hyperref}
\usepackage{mathptmx}
\ifFull
\fi

\let\oldendproof\endproof
\renewcommand{\endproof}{\qed\oldendproof}

\ifFull
\newtheorem{theorem}{Theorem}[section]
\newtheorem{lemma}[theorem]{Lemma}

\newtheorem{definition}[theorem]{Definition}
\fi

\newcommand{\R}{{\bf R}}
\newcommand{\Z}{{\bf Z}}
\newcommand{\floor}[1]{{\left\lfloor #1 \right\rfloor}}
\newcommand{\abs}[1]{{\left\lvert #1 \right\rvert}}

\ifFull
\title{Priority Range Trees}
\author{Michael T. Goodrich \\
Dept. of Computer Science \\
Univ. of California, Irvine \\
\url{http://www.ics.uci.edu/~goodrich/}
\and 
Darren Strash \\
Dept. of Computer Science \\
Univ. of California, Irvine \\
\url{http://www.ics.uci.edu/~dstrash/}
}
\date{}
\else
\title{Priority Range Trees}
\author{Michael T. Goodrich and Darren Strash}
\institute{Department of Computer Science, University of California, Irvine, USA}
\fi

\begin{document}

\maketitle

\begin{abstract}
We describe a data structure, called a \emph{priority range tree},
which accommodates fast orthogonal range reporting queries on
prioritized points.
Let $S$ be a set of $n$ points in the plane, where each
point $p$ in $S$ is assigned a weight $w(p)$ that is polynomial in $n$,
and define the rank of $p$ to be $r(p)=\lfloor \log w(p) \rfloor$. Then the priority
range tree can be used to report all points in a three- or four-sided 
query range $R$ with rank at least $\lfloor \log w \rfloor$ in time $O(\log W/w + k)$,
and report $k$ highest-rank points in $R$ in time
$O(\log\log n + \log W/w' + k)$, where $W=\sum_{p\in S}{w(p)}$,
$w'$ is the smallest weight of any point reported, and $k$ is the
output size. All times assume the standard RAM model of computation.
If the query range of interest is three sided, then the priority 
range tree occupies $O(n)$ space, otherwise $O(n\log n)$ space is used to 
answer four-sided queries. 
These queries are motivated by the Weber--Fechner Law,
which states that humans perceive and interpret data
on a logarithmic scale. 
\end{abstract}

\section{Introduction}
Range searching is a classic problem that has received much attention
in the Computational Geometry literature (e.g.,
see~\cite{a-rs-97,ae-rsir-99,b-mdc-80,c-fsnaq-86,cg-fc1ds-86,%
dhm-brt-09,fmnt-ldsts-86,l-dsorq-78,m-pst-85,o-edsrs-88}).
In what is perhaps the simplest form of range searching, called 
\emph{orthogonal range reporting}, we are given a rectangular,
axis-aligned query range $R$ and our goal is to report 
the points $p$ contained inside $R$ for a given point set, $S$.

A recent challenge with respect to the deployment and use of range
reporting data structures, however, is that modern data sets
can be massive and the responses to typical queries can be overwhelming.
For example, at the time of this writing, a Google query for
``\texttt{range search}'' results in approximately 363,000,000 hits!
Dealing with this many responses to a query is clearly beyond the capacity of
any individual.

Fortunately, there is a way to deal with this type of
information overload---\emph{prioritize} the data and return
responses in an order that reflects these priorities.
Indeed, the success of the Google search engine is largely due to the
effectiveness of its PageRank 
prioritization scheme~\cite{sp-alshw-98,pbmw-pcrbo-99}.
Motivated by this success, our interest in this paper is on the design 
of data structures that can use similar 
types of data priorities to organize the
results of range queries.

An obvious solution, of course, is to treat priority as a
dimension and use existing higher-dimensional range searching techniques
to answer such three-dimensional queries (e.g., see~\cite{a-rs-97,ae-rsir-99}).
However, this added dimension comes at a cost, in that it 
either requires a logarithmic slowdown in query time or an increase in
the storage costs in order to obtain a logarithmic query time~\cite{alstrup2000}. 
Thus, we are interested in prioritized range-searching solutions that can take
advantage of the nature of prioritized data to avoid viewing
priority as yet another dimension.

\subsection{The Weber-Fechner Law and Zipf's Law}
Since data priority is essentially a perception, it is appropriate to apply
perceptual theory to classify it.
Observed already in the 19th century, in what has come to be known as
the \emph{Weber--Fechner Law}~\cite{Dehaene2003145,h-vdiwf-24}, 
Weber and Fechner observed that, in many instances, there
is a logarithmic relationship between stimulus and perception. 
Indeed, this relationship is borne out in several 
real-world prioritization schemes.

For instance, Hellenistic astronomers used their eyesight
to classify stars into six levels of brightness~\cite{e-hpaa-98}.
Unknown to them, light intensity is perceived by the human eye
on a logarithmic scale, and therefore their brightness levels 
differ by a constant factor. 
Today, star brightness, which is referred to as 
\emph{apparent magnitude}, is still measured on a logarithmic scale~\cite{enc-of-astro}. Furthermore, 
the distribution of stars
according apparent magnitude follows an exponential scale 
(see Fig.~\ref{fig:plots}(a)--(b)). 

\begin{figure}[!t]
\begin{center}
\begin{tabular}{cc}
\includegraphics[width=2.3in]{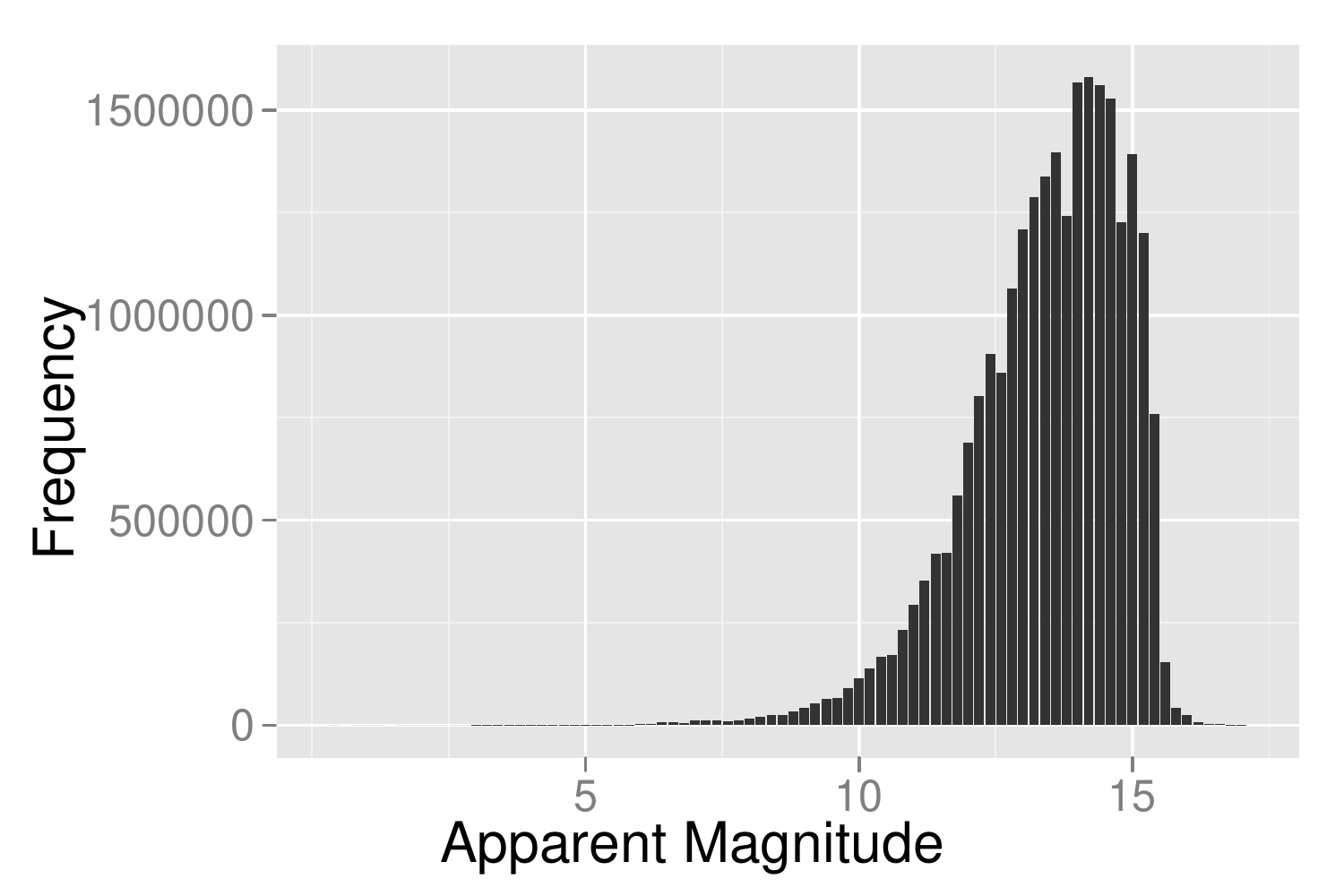} &
\includegraphics[width=2.3in]{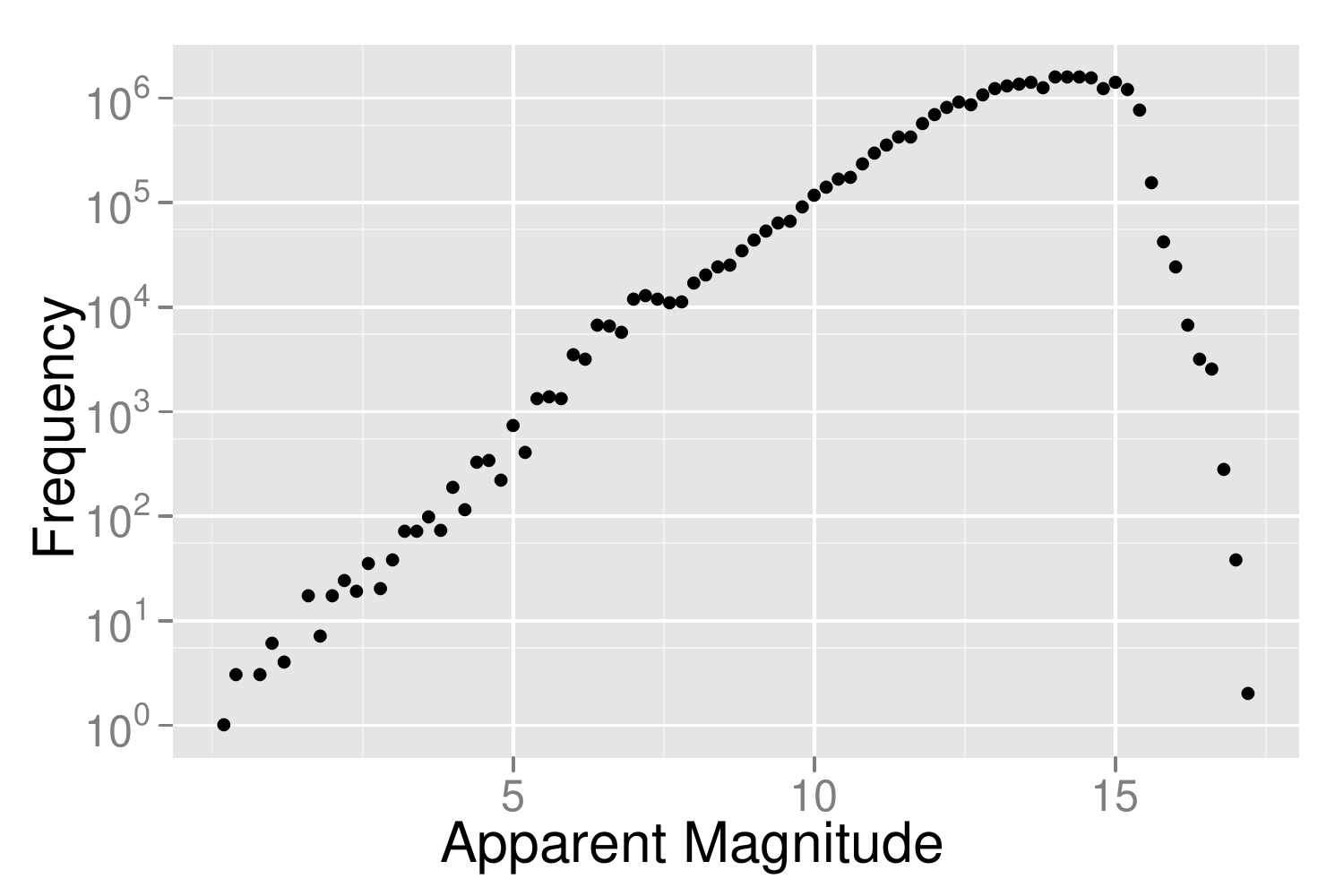}\\
(a)&(b) \\
\includegraphics[width=2.3in]{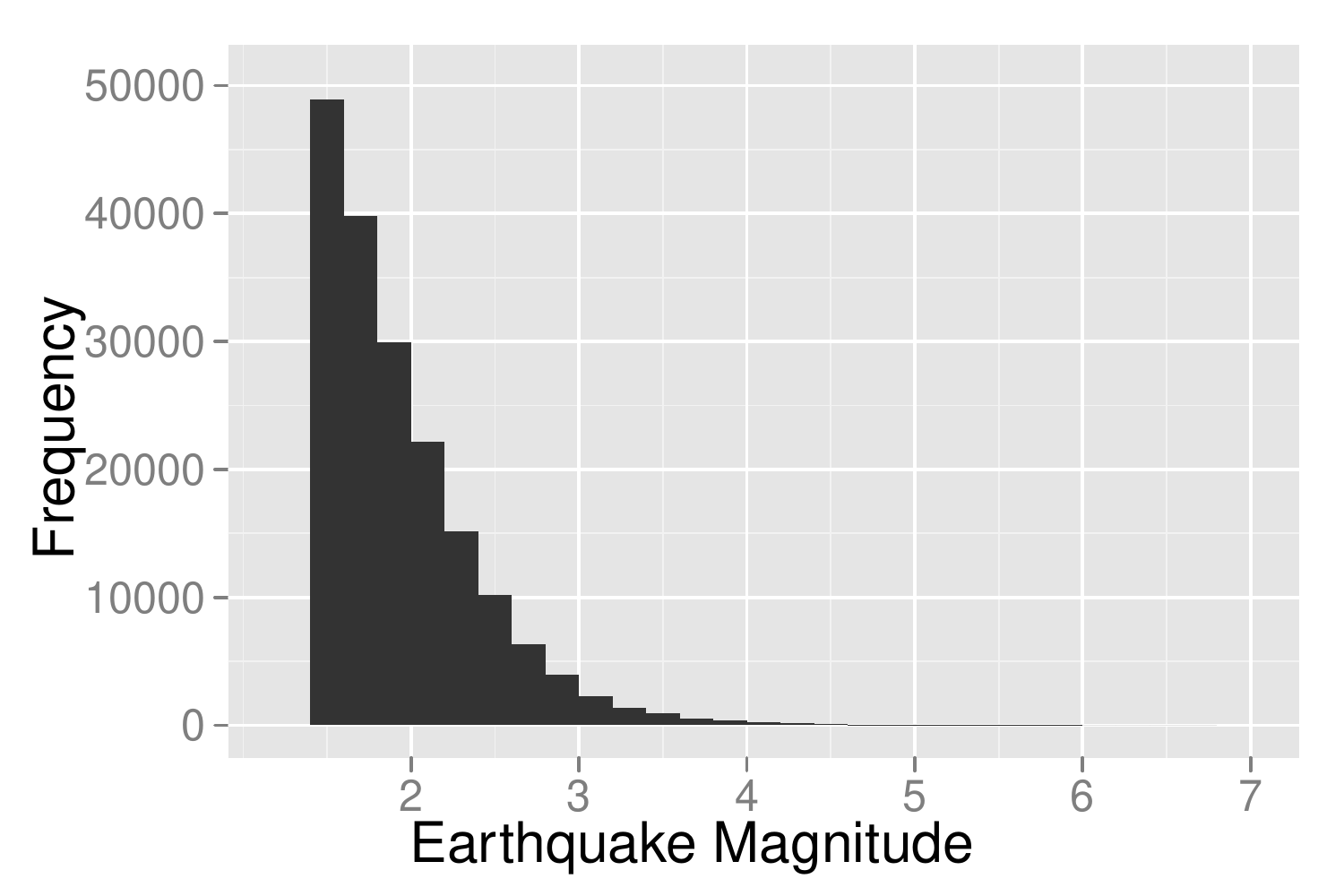} &
\includegraphics[width=2.3in]{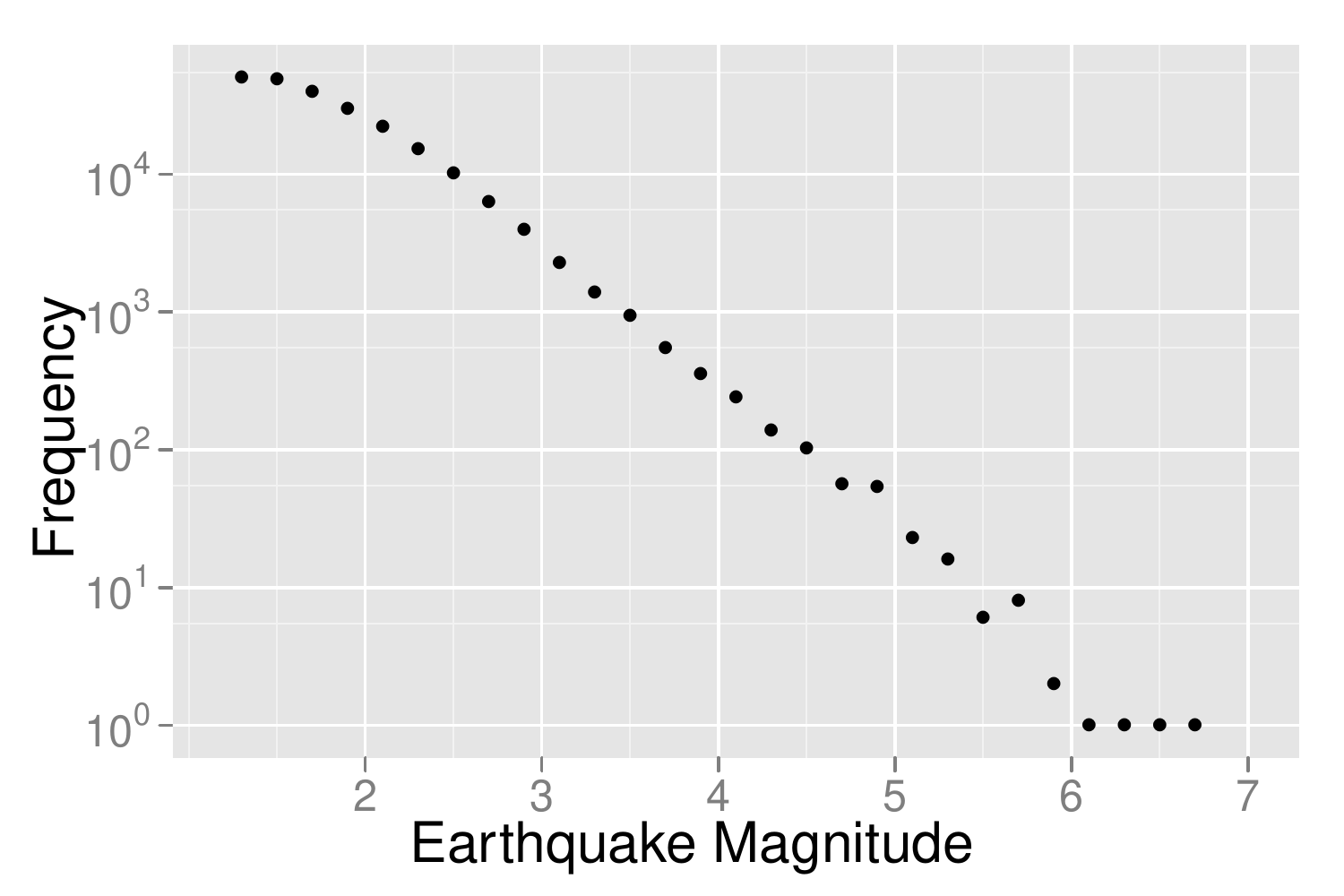}\\
(c)&(d)
\end{tabular}
\caption{(a) Frequency of celestial bodies by apparent magnitude 
for 25 million celestial bodies in the Guide Star Catalog (GSC) 1.2~\cite{GSC1-2}, 
(b) plotted on a log-linear scale. 
(c) Frequency of earthquakes by Richter magnitude for 20 years of 
California earthquakes, (d) plotted on a log-linear scale. Note that because
the measurements are made on a logarithmic scale, the straight line in the 
log-linear plots implies that there is a power law at work.}
\label{fig:plots}
\end{center}
\end{figure}

Logarithmic scale measurements are not confined to the intensity 
of celestial bodies, however. 
Charles Richter's scale~\cite{b-rsdud-89,r-es-58} 
for measuring earthquake magnitude is also logarithmic, 
and earthquake magnitude frequency follows an exponential distribution as well (see Fig.~\ref{fig:plots}(c)--(d)). 
Moreover, as with astrophysical objects, 
range searching on geographic coordinates for 
earthquake occurrences is a common scientific query.

Similar in spirit to the Weber-Fechner Law, Zipf's Law is an
empirical statement about frequencies in real-world data sets. 
Zipf's Law (e.g., see~\cite{fb-irdsa-92})
states that the frequency of a data value in real world
data sets, such as words in documents, is inversely 
proportional to its rank. In other words, 
the relationship between frequency and rank follows a power law.
For example, the popularity of web pages on the Internet
follows such a distribution~\cite{bcl-cws-05}.

\subsection{Problem Statement}
Conventional range searching has no notion of priority. 
All points are considered equal and are dealt with equally. 
Nevertheless, 
as demonstrated by the Weber-Fechner Law and Zipf's Law,
there are many real-world applications where data points are not
created equal---they are prioritized. We therefore aim to 
develop range query data structures that handle these priorities
directly. 
Specifically, we seek to take advantage of prioritization in two ways:
we would like
query time to vary according to the priority of items affecting the query,
and we want to allow for 
items beyond a priority threshold not to be involved in a
given query.

Because of the above-mentioned laws, we feel we can safely
sacrifice some granularity in item weight, focusing instead on 
their logarithm, to fulfill these goals. 
The ultimate design 
goal, of course, is that we desire data structures that provide
meaningful prioritized responses but do not suffer the logarithmic 
slowdown or an increase in space that would come from treating
priority as a full-fledged data dimension. 
To that end, given an item $x$, let us assume that it is given a
priority, $p(x)$, that is positively correlated to $x$'s importance.
So as to normalize item importance,
if such priorities are already defined on a logarithmic scale (like the
Richter scale), then we
define $x$'s \emph{rank}, $r(x)$,  as $r(x)=\lfloor p(x)\rfloor$
and we define $x$'s \emph{weight}, $w(x)$, as $w(x)=2^{p(x)}$.
Otherwise, if priorities
are defined on a uniform scale (like hyperlink in-degree on
the World-wide web), then we define
the $w(x)=p(x)$ and we define $r(x) = \floor{\log{w(x)}}$. We
further assume that weight is polynomial in the number of inputs. This 
assumption implies that there are $O(\log n)$ possible ranks, and 
that $\log W/w = O(\log n)$, where $w$ is a weight polynomial in $n$ 
and $W$ is the sum of $n$ such weights. 
Given these normalized definitions of rank and weight,
we desire efficient data structures that can support the following 
types of prioritized range queries:

\begin{itemize}
\item
\emph{Threshold query}:
Given a query range, $R$, and a weight, $w$, report 
the points in $R$ with rank greater than or equal to $\floor{\log{w}}$. 
\item
\emph{Top-$k$ query}:
Given a query range $R$ and an integer, $k$, report 
the top $k$ points in $R$ based on rank. 
\end{itemize}

\subsection{Prior Work}

As mentioned above, range reporting data structures are well-studied in the Computational
Geometry literature (e.g., see the excellent surveys by Agarwal~\cite{a-rs-97}
and Agarwal and Erickson~\cite{ae-rsir-99}).
In $\R^2$, 2- and 3-sided range queries can
be answered optimally using McCreight's priority search tree~\cite{m-pst-85}, which
uses $O(n)$ space and $O(\log n + k)$ query time.
Using the range trees of Bentley~\cite{b-mdc-80}, and the
fractional cascading technique of Chazelle and Guibas~\cite{cg-fc1ds-86}, 4-sided 
queries can be answered using $O(n\log n)$ space and $O(\log n + k)$ time.
In the RAM model of computation, 4-sided queries can be answered using 
$O(n\log^{\epsilon} n)$ space and $O(\log n + k)$ time~\cite{chazelle88}. 
Alstrup, Brodal, and Reuhe~\cite{alstrup2000} further showed that 
range reporting in $\R^3$ can be done with $O(n\log^{1+\epsilon} n)$ space 
and $O(\log n + k)$ query time in the RAM model.


More recently, Dujmovi\'c, Howat, and Morin~\cite{dhm-brt-09}
developed a data structure called a biased range tree which, assuming
that 2-sided ranges are drawn from a probability distribution, can 
perform 2-sided range counting queries efficiently. Afshani, Barbay, and Chan~\cite{afshani2009}
generalized this result, showing the existence of many instance-optimal algorithms.
Their methods can be viewed as solving an orthogonal problem
to the one studied here, in other words,
in that their points have no inherent weights in and of themselves
and it is the distribution of ranges that determines their importance.

\subsection{Our Results}

Given a set $S$ of $n$ points in the plane, where each point $p$
in $S$ is assigned a weight $w(p)$ that is polynomial in $n$, we provide a
data structure, called a \emph{priority range tree}, which accommodates fast 
three-sided orthogonal range reporting queries. In particular,
given a three-sided query range $R$ and a weight $w$, our data structure can
be used to answer a threshold query,
reporting all points $p$ in $R$ such that $\lfloor \log{w(p)} \rfloor \geq 
\lfloor \log{w} \rfloor$ in time $O(\log W/w + k)$, where $W$ is the sum of the 
weights of all points in $S$.
In addition, we can also support top-$k$ queries, reporting $k$ points
that have the highest $\lfloor \log{w(.)} \rfloor$ value in $R$
in time $O(\log \log n + \log W/w + k)$, where $w$ is the smallest weight
among the reported points. The priority range tree data structure occupies
linear space, and operates under the standard RAM model of computation.
Then, with a well-known technique for converting a 3-sided range
reporting structure into a 4-sided range reporting structure, we show how to construct a data
structure for answering prioritized 4-sided range queries with
similar running times to those for our 3-sided query structure.
The space for our 4-sided query data structure is larger by a
logarithmic factor.

\subsection{A Note About Distributions}
A key feature of the priority range tree is that it is distribution agnostic.
This distinction is crucial, since if the distribution of priorities is
fixed to be exponential, then a trivial data structure achieves the same results:
for i = $1$ to $\floor{\log w_{\text{max}}}$, create a priority search tree
$P_i$ containing all points with weight $2^{i}$ and above. Because the distribution is
exponential, the number of elements in $P_{i}$ is at most
$n/{2^i} \leq W/{2^i}$, and hence, querying $P_{\floor{\log w}}$ correctly 
answers the query and takes time $O(\log W/w + k)$. 
Furthermore, the space used for all data structures is $\sum_i n/2^i = O(n)$. 

However, such a strategy does not work for other distributions (including
power law distributions, which commonly occur in practice) since the 
storage for each data structure becomes too great to meet the desired 
query time and maintain linear space. Thus, our data structure provides query 
times approaching the information theoretic lower bound, in linear space, 
regardless of the distribution of the priorities.

\section{Preliminary Data Structuring Techniques}
In this section, we present some techniques that we use
to build up our priority range tree data structure.

\subsection{Weight-balanced Priority Search Trees}
\label{sec:one-dim}
Consider the following one-dimensional range reporting problem:
Given a set $S$ of $n$ points in $\R$, where each point $p$ has weight $w(p)$,
we would like to preprocess $S$ into a data structure so that we can report all 
points in the query interval $[a,b]$ with weight greater or equal to
$w$. Storing the points in a priority search tree~\cite{m-pst-85}, affords
$O(\log{n} + k)$ query time using linear space. 
We can obtain a query time of $O(\log W/w + k)$ by ensuring that the priority
search tree is \emph{weight balanced}.

\begin{definition} We say a tree is weight balanced if item $i$ with
weight $w_i$ is stored at depth $O(\log W/w_i)$, where $W$ is the sum of the weights
of all items stored in the tree.
\end{definition}

To build this search tree, we use a trick similar to Mehlhorn's rule 2~\cite{mehlhorn75f}.
We first choose the item with the highest weight to be stored at the root.  
We then divide the remaining points into two sets $A$ and $B$ such
that the $x$-value of every point in $A$ is less than or equal to the
$x$-value of every point in $B$ and $\abs{\sum_{a\in A}w(a) -\sum_{b\in B}w(b)}$
is minimized. Finally, we store the maximum $x$ value from $A$ in the root to facilitate
searching, and then recursively build the left and right subtrees
on sets $A$ and $B$. We call this technique 
\emph{split by weight}. Priority search trees are 
built much the same way, except that $A$ and $B$ are chosen to have 
approximately the same cardinality, which we 
call \emph{split by size}. 

The resulting search tree is both weight balanced and heap ordered by weight, 
and can therefore be used to answer range reporting queries with the 
same procedure as the priority search tree. 

\begin{lemma} The weight-balanced priority search tree consumes $O(n)$ space,
and can be used to report all points in a query range $[a,b]$ with weight at least $w$ in 
time $O(\log(W/w) + k)$ where $W$ is the sum of the weights of all points in the tree.
\label{theorem-wb}
\end{lemma}

\subsection{Persistent Heaps}

\begin{figure}[!t]
\begin{center}
\includegraphics[width=3.0in]{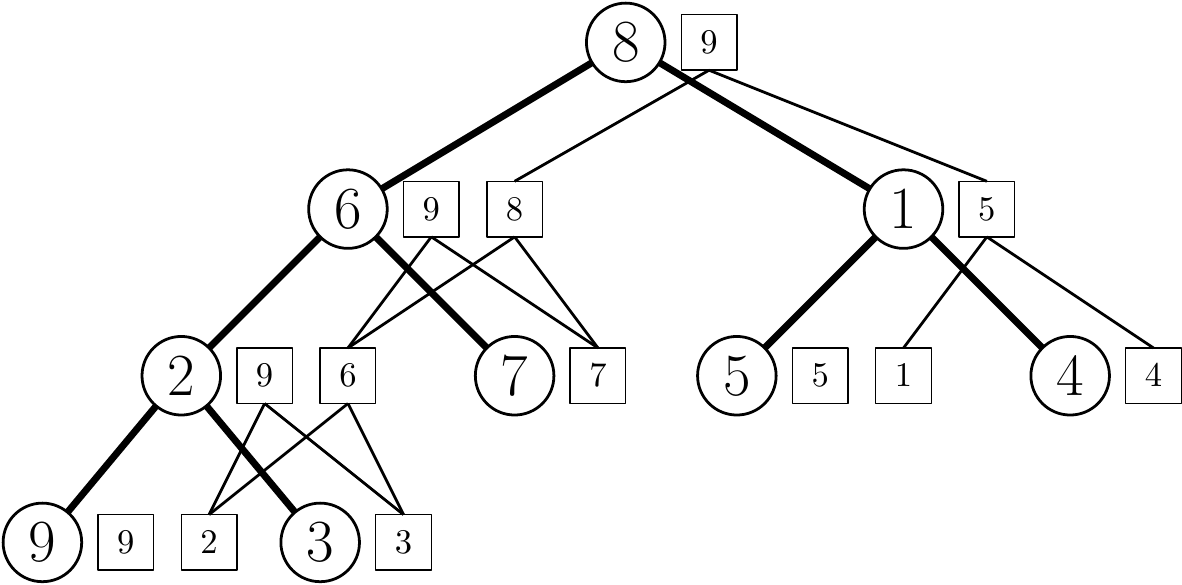}
\caption{A persistent heap. Circular nodes represent the original 
binary tree and square nodes represent heaps at each step
of the BuildHeap algorithm.}
\label{fig:hist-pst}
\end{center}
\end{figure}

The well-known BuildHeap algorithm can transform any complete binary tree 
into a heap in linear time~\cite{floyd64}. Using the node-copying method for
making data structures persistent~\cite{driscoll89}, we can maintain a record of 
the heap as it exists during each step of the BuildHeap algorithm,
allowing us to store a heap on every subtree in linear space.
We call this data structure a \emph{persistent heap}.

\begin{lemma} Let $T$ be a tree with $n$ nodes. If the BuildHeap algorithm runs in time $O(n)$
on $T$, then we can augment every node of $T$
with a heap of the elements in its subtree
using extra space $O(n)$.
\end{lemma}
\begin{proof}
For each swap operation of the BuildHeap algorithm, we do not swap within
the tree, but we create two extra nodes, add the swapped elements
to these nodes, and add links to the heaps from 
the previous stages of the algorithm (see Fig.~\ref{fig:hist-pst}). 
\end{proof}

Given $n$ points in $\R^2$, this strategy can be used as a substitute
for the priority search tree, by first building complete binary search tree
on the $x$ values, and then building a persistent heap using the $y$
values as keys.

\subsection{Layers of Maxima}
We now turn our attention to the following two problems: 
Given $l$ points in $\Z_m \times \R$, preprocess the points into a data structure, 
to quickly answer the following queries.
\begin{enumerate}
\item \emph{Domination Query}: Given a query point $(x,y)$, report all points $(p_x, p_y)$ such that $x\leq p_x$ and $y \leq p_y$.
\item \emph{Maximization Query}: Given a query value $y$, report a point with the largest $x$-value, such that its $y$-value 
is greater than $y$.
\end{enumerate}
This problem can be solved optimally using two techniques: we form the 
layers of maxima of the points and use fractional cascading to reduce
search time.

A point $p \in \R^2$, \emph{dominates} a point $q \in \R^2$ iff each
coordinate of $p$ is greater than that of $q$. A point $p$ is said 
to be a \emph{maximum} of a set $S$ iff no point in $S$ dominates $p$.
Given a set $S$, if we find the set of maxima of $S$, remove the maxima 
and repeat, then the resulting sets of points are called the 
\emph{layers of maxima} of $S$. 

\begin{figure}[!t]
\begin{center}
\begin{tabular}{c@{\hspace*{0.75cm}}c}
\includegraphics[width=2in]{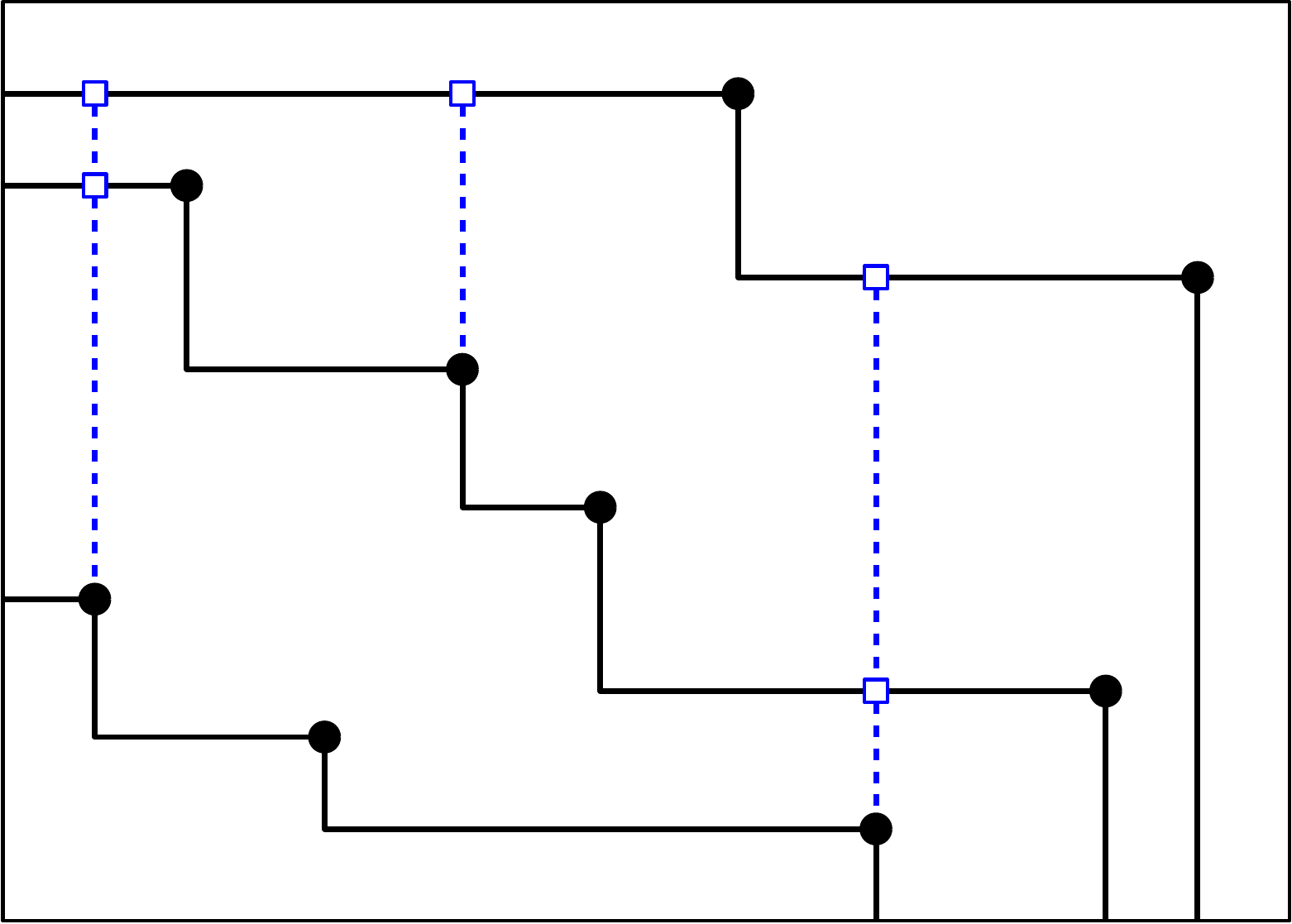}&
\includegraphics[width=2in]{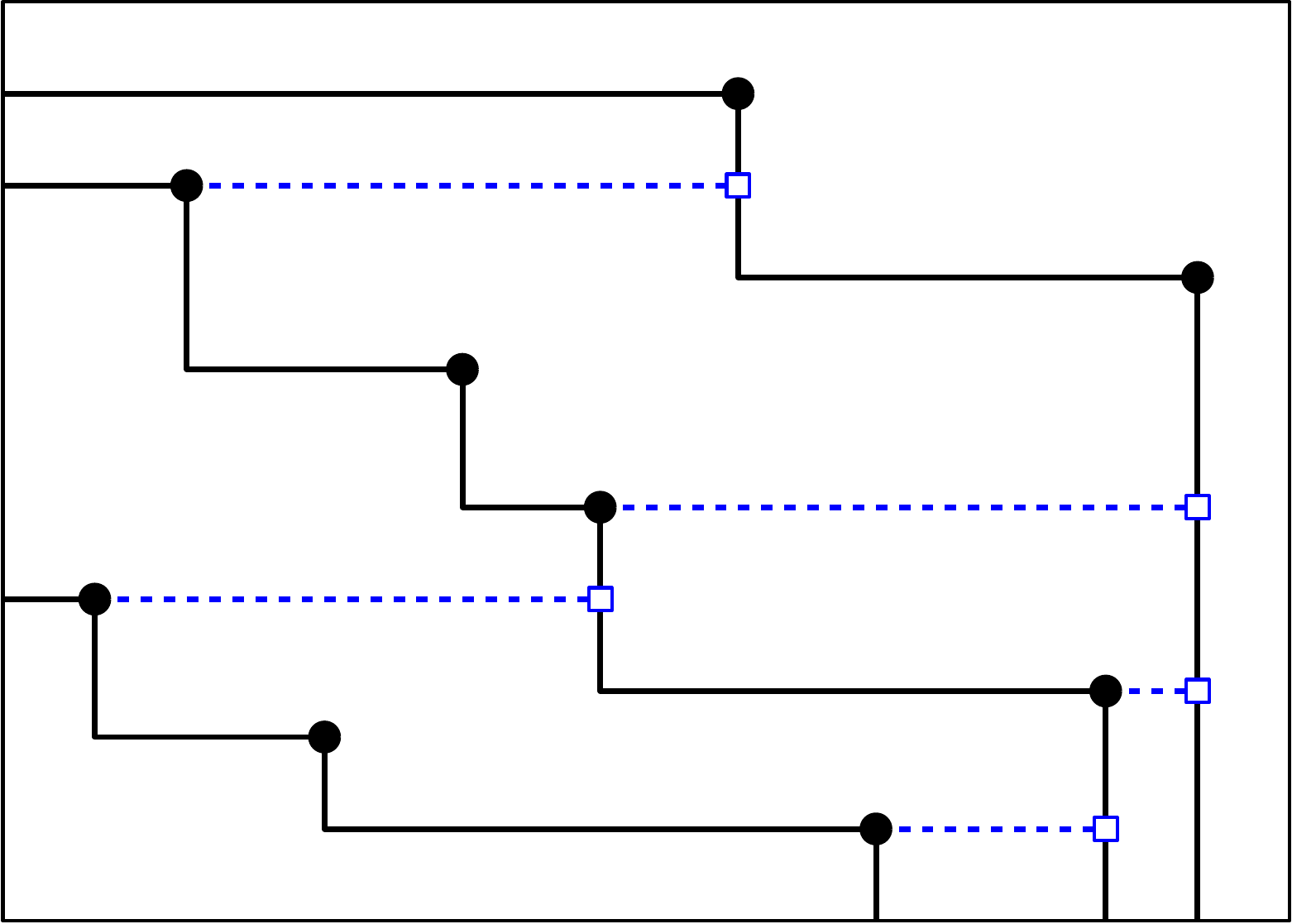}\\
(a)&(b)
\end{tabular}
\caption{The layers of maxima with fractional cascading (a) from
bottom to top and (b) from left to right.}
\label{fig:layers}
\end{center}
\end{figure}

We begin by constructing the layers of maxima of 
the $l$ points. We then form a graph from the layers of maxima by creating a vertex
for each point and connecting vertices that are in the same layer
in order by $x$ coordinate.

We first fractionally cascade the points from bottom to top, sending up
every other point from one layer to the next, including points copied from
previous layers (see Fig.~\ref{fig:layers}(a)). We then repeat the same
procedure, fractionally cascading the points from left to right 
(see Fig.~\ref{fig:layers}(b)).
For bottom-to-top fractional cascading, we create $m$ entry 
points into the top layer our data structure stored as an array indexed by $x$ value.
Each entry point stores one pointer to the maxima in the top layer that succeeds it in
$x$-value.

To answer domination queries, we enter the catalog at index $x$,
reports all points on the current layer that match the query, 
then jump down to the next layer and repeat.
Each answer can be found with a constant amount of searching. Therefore,
the query takes time $O(\max\{k, 1\})$ where $k$ is the output size. 

To answer maximization queries, we create a catalog of $O(l)$ entry 
points on the right. Each entry point stores a pointer to a point 
(copied or not) on the top layer of maxima. Since the domain of 
$y$ is not constrained to the integers, we perform a search for 
our query $y$ among the entry points and immediately jump to our answer. 
This data structure gives us $O(l)$ space and $O(\log{l})$ query time.

We now have all the machinery to discuss the priority range tree
data structure.

\section{The Priority Range Tree}

In this section, we present a data structure for three-sided
range queries on prioritized points. We assume that each point $p$
has a weight $w(p)$, and we define $r(p)=\floor{\log w(p)}$ to be the rank
of $p$. Given a 
range $R = [x_1, x_2] \times [y, \infty)$ our data structure
accommodates the following queries.

\begin{enumerate}
\item \emph{Threshold Queries}: Given a query weight $w$, report all points $p$ in $R$
whose weight satisfies $\floor{\log{w(p)}} \geq \floor{\log w}$.
\item \emph{Top-$k$ Queries}: Given an integer $k$, report the $k$ points in $R$ with
the highest $\floor{\log{w(\cdot)}}$ value.
\end{enumerate}

We first describe a data structure that has significant storage
requirements, to illustrate how to perform each query.
We then show how to reduce the space requirements.

For our underlying data structure, we build up a weight-balanced
priority search tree on the $x$-values of our points.
On top of this tree, we build one persistent heap for each different 
rank. That is, given rank $r$, we build a persistent heap on the 
$y$-values of points that have rank $r$.
Of course, points with different rank must be compared in this scheme, 
therefore, when building a persistent heap for rank $r$, we treat 
points with rank not equal to $r$ as dummies with $y$-value $-\infty$. 
Once we are done building up these persistent heaps, 
each node has $O(\log{n})$ heaps rooted at it, one for each rank. 
For each node, we store a catalog, which is an array of roots of each of the 
$O(\log n)$ heaps, indexed by rank.
On top of each catalog, we build the fractionally-cascaded layers-of-maxima 
data structure, described in the previous section, storing a
coordinate (\emph{rank}, \emph{y-value}) for each root of the $O(\log n)$ heaps.

\subsection{Threshold Queries} 

We first search for $x_1$ and $x_2$ down to depth 
$O(\log W/w)$ in our weight-balanced priority search tree, checking each point 
for membership as we make our way down the tree. 

Each node on the \ifFull\else search paths to $x_1$ and $x_2$ \fi may 
have left or right subtree whose $x$ values
are entirely in the range $[x_1,x_2]$. For each such subtree, we query 
the layers-of-maxima data structure to find all points in the catalog 
that dominate $(\floor{\log w}, y)$.
If any points are returned, then we perform a layer-by-layer search
through the heaps stored for each rank. We return any points that
satisfy the query $y$ value.

Each layers-of-maxima search can be charged to the search depth
or an answer, and each search within a heap can be charged to
an answer. Therefore, we get the desired running time of $O(\log W/w + k)$.

\subsection{Top-$k$ Queries}
This query type is slightly more involved, so we begin by describing
how to find one point of maximum rank in a query range. 

\paragraph{Max-reporting.}
Given a range $R = [x_1, x_2] \times [y, \infty)$, the \emph{max-reporting}
problem is to report one point in $R$ with maximum rank. 

A first attempt is to search for $x_1$ and $x_2$,
and run a maximization query for each layers-of-maxima data structure along
the search path, maintaining the point in $R$ with maximum rank found so far. 
Although this is a correct algorithm, there are two issues with this approach, which are brought
about because we do not have a query weight:  

\begin{enumerate}
\item The search may reach depth $\omega(\log W/w')$, where $w'$ is the weight of the answer. 

\item Each query to a layers-of-maxima data structure takes time $O(\log\log n)$.
\end{enumerate}

Therefore, we maintain a depth limit, initially $\infty$, 
telling our algorithm when to stop searching. If there is no
point in $R$, then our search reaches a leaf at depth $O(\log n)$
and stops. If we find a point $p$ in $R$, we 
decrease the depth limit to $c\log W/w(p)$, where $c$ is the constant
hidden in the big-oh notation for the weight-balanced priority
search tree. If our search reaches the depth limit, then there
are no points with greater rank lower in the tree, and we
can stop searching. Otherwise, every time we encounter point 
in $p'$ in $R$ with higher rank $r'=\floor{\log w'}$
we decrease the depth limit to $c\log W/w'$. 

We reduce the layers-of-maxima query time by fractionally cascading
the layers-of-maxima data structure across the entire search tree,
allowing us to do one $O(\log\log n)$-time query in the root catalog,
and $O(1)$ extra work in the catalogs on the search path.

With these two changes, the search takes time $O(\log\log n + \log W/w')$ total
($O(\log n)$ if no such point exists).

\paragraph{Top-$k$ Reporting}

We now extend the max-reporting algorithm to report $k$ points 
with highest rank in time $O(\log \log n + \log W/w' + k)$ under 
the standard RAM model.

If the top $k$ points all have the same rank, then we can use
our max-reporting algorithm to find the point with highest
rank, and use the threshold queries to recover all $k$ points.
However, if we have to find points with lower rank, we want 
to avoid doing an expensive search for each rank. We can 
accomplish this goal with a priority queue.

Perform an initial max-reporting search. Every point in $R$ 
encountered during our search is inserted into a priority 
queue with key equal to its rank. 
Along with each point we store a link back to the location in the 
layers-of-maxima data structure where it was found. 
When we finish the initial max-reporting query, we iterate
the following process:

\begin{enumerate}
\item Remove the point $p$ with maximum rank $r$ from the priority queue.
\item Enter the layers-of-maxima data structure at point $p$, and insert both
the predecessor of $p$ on the same layer and on the layer below into the
priority queue. Each one of these points are candidates for reporting.
We then mark points to ensure that duplicates are not added to the priority queue.
\item Search in the heap data structure where point $p$ was found, 
and report any additional points it contains that are in $R$ (without exceeding $k$ points). 
\item If we have reported $k$ points then we are done. Otherwise, look 
at the point with maximum rank in the priority queue.
If its rank $r'=\floor{\log{w'}}$ is less than $r$, then we increase our search depth to $c\log W/w'$,
and continue searching, adding points the priority queue as before.
\end{enumerate}



This priority queue can be efficiently implemented in the standard RAM model.
We store our priority queue as an array $P$, indexed by key.
We store in cell $P[r]$ a linked list of elements with 
key $r$. Additionally, we maintain two values,
$r_{\text{max}}$ and $r_{\text{min}}$, which is the maximum (minimum) key,
of all elements in the priority queue. 
We insert an item with key $k$, by adding it to the linked list $P[k]$ in $O(1)$ time,
and updating $r_{\text{max}}$ and $r_{\text{min}}$. To remove an item with the maximum 
key $k$, we remove it from the linked list $P[k]$. If $P[k]$ becomes empty, then 
we update $r_{\text{max}}$ (and possibly $r_{\text{min}}$).

We now show that our top-$k$ reporting algorithm has running time 
$O(\log \log n + \log W/w' + k)$.
We spend an initial $O(\log \log n)$-time search for our fractional cascading.
Our initial search and subsequent extensions of the search path 
takes time $O(\log W/w')$, by virtue of our depth limit. For each heap
in which we perform a layer-by-layer search, we can charge the search time
to answers reported. All that remains to be shown is that the priority 
queue operations do not take too much time.

For each point discovered in a layers-of-maxima query along the search path, we perform at most one 
insert into the priority queue, thus we do $O(\log W/w')$ of these insertions into the priority queue, each 
one taking constant time. Furthermore, we can charge all of our priority 
queue remove operations (excluding pointer updates) to answers. Updating the
priority queue pointers does not negatively impact the running time, 
since the total number
of array cells we march through to do the pointer updates is $O(\log W/w')$. For each
remove operation, we may perform up to two subsequent insertions, which we can charge to
answers. Therefore, we get a total running time of $O(\log \log n + \log W/w' + k)$.

As described, the priority range tree consumes $O(n\log^2 n)$ 
space, since each persistent heap may consume $O(n\log n)$ space
and we store $O(\log n)$ such persistent heaps.\footnote{Each persistent
heap uses space $O(n\log n)$ instead of $O(n)$ because there is
no guarantee that BuildHeap will run in linear time on our underlying tree.}

\subsection{Reducing the Space Requirements}
We can reduce the space to $O(n)$ by making several modifications to the search tree.
We build our underlying tree using the split-by-weight strategy down to depth
$\frac{1}{2}\log{n}$, then switch to a split-by-size strategy for the deeper 
elements, forcing these split-be-size subtrees to be complete. By 
switching strategies we do not lose the special
properties that we desire: the tree is still weight balanced and heap 
ordered on weights. Finally, we do not store auxiliary data structures for 
split-by-size subtrees with between $\frac{1}{2}\log{n}$ 
and $2\log{n}$ elements in them, we only store one catalog for
each such subtree. We call these subtrees \emph{buckets}.


\begin{lemma}
The priority range tree consumes $O(n)$ space.
\end{lemma}
\begin{proof}
Each layers-of-maxima data structure uses $O(\log n)$ storage.
We store one such data structure with each split-by-weight node
and one for each split-by-size node excepting those
in subtrees with between $1/2\log n$ and $2\log n$ elements.
There are $O(\sqrt{n})$ split-by-weight nodes, since they are all 
above depth $1/2\log n$. Each subtree $T_i$ rooted at depth $1/2\log n$
of size $n_i$ will have at most $2n_i/\log n$ nodes that
store the auxiliary data structure. Therefore, the total space 
for layers-of-maxima data structures is
$O(\sqrt n\log n) + \sum_iO(n_i/\log n)O(\log n) = O(n)$.

For the persistent heaps, we ensure that each subtree $T_i$
rooted at depth $1/2\log n$ is complete, on which we know BuildHeap 
will run in linear time. If $T_i$ contains fewer than $1/2\log n$ nodes,
then $T_i$ has $O(\log n)$ heap nodes. Otherwise, each subtree $T_i$ 
stores $O(n_i/\log n)$ heap nodes per rank. 

Each node above depth $1/2\log n$ can contribute $O(\log{n})$ swaps for each heap.
Each heap therefore requires $O(\sqrt{n}\log n + \sum_i n_i/\log{n}) = O(n/\log{n})$ space. Since we have $O(\log{n})$ heaps, our data structure requires
$O(n)$ space.
\end{proof}

These structural changes affect our query procedures. In particular,
there are three instances where the bucketing affects the query:
\begin{enumerate}
\item If the initial search phase hits a bucket, then we
exhaustively test the $O(\log n)$ points in the bucket. 
We only reach a bucket if the search path is of length $\Omega(\log n)$
and therefore we can amortize this exhaustive testing over the search. 

\item If we reach a bucket during our layer-by-layer search through a heap,
then exhaustively searching through the bucket is not an option, as this
would take too much time. Instead, we augment the layers-of-maxima data structure 
for the bucket so we can walk through lower $y$ values with the same rank
as the maxima (possibly producing duplicates). 

\item By the very nature of the BuildHeap algorithm, it is possible 
that points in $R$ were DownHeaped
into the buckets, and that information is also lost. Therefore, when looking
layer by layer through the heaps we also need to test for membership of
the tree nodes in addition to the heap nodes (also possibly producing duplicates).

\end{enumerate}

Each point matching the query will be encountered at most three times,
once in each search phase. Therefore, we can avoid returning duplicates 
with a simple marking scheme without increasing the reporting
time by more than a constant factor.

\begin{theorem}
The priority range tree consumes $O(n)$ space and can be used to answer 
three-sided threshold reporting queries with rank above $\floor{\log w}$ in 
time $O(\log W/w + k)$, and top-$k$ reporting queries in time 
$O(\log\log n + \log W/w' + k)$, where $W$ is the sum of the weights
of all points in $S$, $w'$ is the smallest weight of the reported 
points, and $k$ is the number of points returned by the query.
\end{theorem}

\section{Four-sided Range Reporting}
Our techniques can be extended to four-sided range queries with an added
logarithmic factor in space using a twist on a well-known transformation. 
Given a set $S$ of points in $\R^2$, we form a weight-balanced 
binary search tree on the $x$-values
of the points, such that the points are stored in leaves (e.g., a biased 
search tree~\cite{bst-bst-85}). For each internal node, we store 
the range of $x$-values of points contained in its subtree. 
For each internal node $x$ (except the root), we store a priority range tree $P_x$
on all points in the subtree: 
if $x$ is a left child then $P_x$ answer queries of the form 
$[a, \infty) \times [c,d]$, if $x$ is a right child then $P_x$
answers queries of the form $(-\infty,b] \times [c, d]$.
We can answer four-sided query given the range $[a,b] \times [c,d]$, by doing the following:

We first search for $a$ and $b$ in in the weight-balanced binary search tree. 
Let $s$ be the node where the search for $a$ and $b$ diverges, then
$a$ must be in $s$'s left subtree $\text{left}(s)$ and $b$ must be in $s$'s 
right subtree $\text{right}(s)$. Then we query $P_{\text{left}(s)}$ for points
in $[a, \infty) \times [c, d]$ and query $P_{\text{right}(s)}$ for points in 
$(-\infty, b] \times [c, d]$ and merge the results. In the case of top-$k$ queries, 
we must carefully coordinate the search and reporting in each priority range 
tree, otherwise we may search too deeply in one of the trees or report 
incorrect points.

\section{Conclusion}
Our priority range tree data structure can be used to report points in a three-sided 
range with rank greater than or equal to $\floor{\log w}$ in time $O(\log W/w + k)$,
and report the top $k$ points in time $O(\log\log n + \log W/w' + k)$, where
$w'$ is the smallest weight of the reported points, using
linear space. These query times extend to four-sided ranges with 
a logarithmic factor overhead in space. Our results are possible because of 
our reasonable assumptions that the weights are polynomial in $n$, and
that the magnitude of the weights, rather than the specific weights themselves,
are more important to our queries. 

\subsubsection*{Acknowledgments.}
This work was supported in part by NSF grants 0724806, 0713046, 0830403, and ONR grant N00014-08-1-1015.

{
}

\end{document}